\documentclass[12pt]{amsart}
\usepackage{inconsistent_beliefs_BE}

\title{}

\begin{document}

\vspace*{3ex minus 1ex}
\begin{center}
\Large \textsc{Existence of Bayesian Equilibria in Incomplete Information Games without Common Priors}
\bigskip
\end{center}

\date{%
%TCIMACRO{\TeXButton{Today}{\today}}%
%BeginExpansion
\today%
%EndExpansion
}

\vspace*{3ex minus 1ex}
\begin{center}
Denis Kojevnikov\,\orcidlink{0009-0008-2514-5835} and Kyungchul Song\\
\textit{Tilburg University and University of British Columbia}
\end{center}

\thanks{
	We are grateful to Li Hao, Vitor Farinha Luz, Humberto Moreira, Paulo Klinger Monteiro, and participants at UBC Theory Lunch Seminar. All errors are ours. Song acknowledges that this research was supported by Social Sciences and Humanities Research Council of Canada. Corresponding address: Denis Kojevnikov, Department of Econometrics and Operations Research, Tilburg University, Netherlands. Email address: \href{mailto:d.kojevnikov@tilburguniversity.edu}{d.kojevnikov@tilburguniversity.edu}.
}

\begin{abstract}
	This paper focuses on finite-player incomplete information games where players may hold mutually inconsistent beliefs without a common prior. We introduce absolute continuity of beliefs, extending the classical notion of absolute continuity of information in \cite{Milgrom/Weber:85:MOR}, and prove that Bayesian equilibria exist in a broad class of games, including those with discontinuous payoffs. Examples illustrate the scope and implications of our findings.

	\medskip


	\medskip

\end{abstract}

\maketitle

\vspace*{5ex minus 1ex}

\section{Introduction}

Studies of incomplete information games often assume that players hold beliefs derived from a common prior. The common prior assumption has proven instrumental in generating several foundational results in economic theory (e.g., \citeay{Aumann:76:AS}, \citeay{Milgrom/Stokey:82:JET}, \citeay{Sebenius/Geanakoplos:83:JASA}, \citeay{Aumann:87:Eca}, \citeay{Aumann/Brandenburger:95:Eca}). It also provides mathematical tractability in many applied settings (\citeay{Myerson:04:MS}). However, critics argue that it is neither plausible nor meaningful as a representation of the epistemic aspects of strategic environments (\citeay{Dekel/Gul:97:AEE}, \citeay{Gul:98:Eca}). (See also \citeay{Morris:95:EP} and \citeay{Bonanno/Nehring:99:IJGT} for a discussion of the common prior assumption.)

Despite the controversies, the common prior assumption remains a fundamental tool in establishing equilibrium existence in incomplete information games. The seminal work of \cite{Milgrom/Weber:85:MOR} showed the existence of a distributional strategy equilibrium in general incomplete information games with payoffs continuous in actions, and \cite{Balder:88:MOR} extended the result by removing the topological requirements on the type sets. More recently, by adapting the framework of \cite{Reny:99:Eca} and conditions in \cite{Monteiro/Page:07:JET}, \cite{Allison/Lepore:14:JET}, and \cite{Prokopovych/Yannelis:14:JME} to an incomplete information setting, \cite{He/Yannelis:16:JET} and \cite{Carbonell-Nicolau/McLean:18:MOR} established the existence of a Bayesian Nash equilibrium in games with discontinuous payoffs. The literature has also explored the existence of a pure strategy equilibrium or an equilibrium with a monotonicity property (\citeay{Athey:01:Eca}, \citeay{McAdams:03:Eca}, and \citeay{Reny:11:Eca}, to name but a few). Except for some restrictive settings (see \citeay{VanZandt:10:JET} for supermodular games and Theorem 10.42 of \citeay{Maschler/Solan/Zamir:20:GameTheory} for games with a finite set of states of the world), the literature on equilibrium existence predominantly adopts the common prior assumption.

In this paper, we consider a general finite-player incomplete information game where the players hold potentially heterogeneous beliefs, not necessarily derived from a common prior.\footnote{
	In econometrics, \cite{Kojevnikov/Song:23:JOE} developed asymptotic inference for large incomplete information games without a common prior when the predictions are based on Bayesian equilibria.
}
It has been noted that the existence of Bayesian equilibria may not be guaranteed if the type space is too large. For instance, using an example from \cite{Hellman:14:JET}, \cite{Friedenberg/Meier:17:ET} constructed a finite-action incomplete information game on the universal type space of \cite{Mertens/Zamir:85:IJGT} that does not have a Bayesian equilibrium. \cite{Levy:24:JME} obtained a positive counterpart allowing heterogeneous but mutually absolutely continuous priors. Specifically, when the state space partitions into commonly known components, equilibrium existence on almost every component transfers to the game as a whole, subject to regularity conditions on the partition.

Our main innovation is the notion of \emph{absolute continuity of beliefs}, which we show is sufficient for the existence of Bayesian equilibria in a wide class of games without a common prior. Absolute continuity of beliefs requires the existence of a product probability measure that dominates each player's belief about the payoff state and the other players' types. When there are countably many types, this condition is satisfied. However, it excludes settings in which a player holds pairwise mutually singular beliefs at uncountably many type values. Such a case arises, for example, when the types are drawn from a continuum and it is common knowledge that their values are identical across the players. As we show in the paper, absolute continuity of beliefs can be viewed as an extension of absolute continuity of information, introduced by \cite{Milgrom/Weber:85:MOR}, to a setting without a common prior.

We study the existence of Bayesian equilibria through the existence of Nash equilibria in a surrogate game. We first characterize the set of Bayesian equilibria as the intersection of the sets of Nash equilibria across a class of surrogate games indexed by product probability measures that dominate the players' beliefs. Our main result shows that a Bayesian equilibrium exists if and only if a Nash equilibrium of the surrogate game exists for some dominating product probability measure. Standard methods apply to the surrogate game because it can be reformulated as a Bayesian game whose common prior is the dominating measure itself, so that absolute continuity of information holds automatically. Thus, our result provides a general approach to establishing the existence of Bayesian equilibria in incomplete information games without a common prior.

Our general framework applies to a wide range of settings, notably incomplete information games with payoff functions that are discontinuous in actions. To illustrate the framework's usefulness, we present an example of incomplete information Cournot games adapted from \cite{Carbonell-Nicolau/McLean:18:MOR}. Applying our results, we show that when the belief map profile satisfies absolute continuity of beliefs and the payoff functions satisfy mild conditions, these games have Bayesian equilibria.

This paper is organized as follows. Section \ref{sec2} introduces a general incomplete information game and the notion of absolute continuity of beliefs. Section \ref{sec3} establishes the characterization and existence results for Bayesian equilibria and presents applications. The appendices contain the omitted proofs and auxiliary results.

\subsubsection*{Preliminaries}
Throughout the paper, $\NN$ and $\R$ denote the sets of natural and real numbers, respectively. For a topological space $X$, $\B(X)$ denotes the Borel $\sigma$-algebra on $X$. A \emph{Polish space} is a separable, completely metrizable topological space, and a \emph{Borel space} is a metrizable space that is homeomorphic to a Borel subset of a Polish space \citep[see, e.g.,][Definition 15.17]{Aliprantis/Border:06:InfiniteDimensionalAnalysis}.\footnote{
	Note that a Borel space is a Borel subset of its completion under a compatible metric \citep[see, e.g.,][Proposition 3.3.7 and Remark 3.3.8]{Srivastava:CourseOnBorelSets}.
}

Let $(X,\X)$ and $(Y,\Y)$ be measurable spaces. The \emph{universal completion} of $\X$ is the $\sigma$-algebra $\bigcap_{\mu} \X^\mu$, where $\mu$ ranges over the probability measures on $(X,\X)$ and $\X^\mu$ denotes the completion of $\X$ with respect to $\mu$. A $\sigma$-algebra is \emph{universally complete} if it equals its universal completion.

We say that a map $f: X \to \R$ is $\X$-measurable if it is $(\X,\B(\R))$-measurable. For measures $\mu$ and $\mu'$ on $(X,\X)$, we write $\mu \ll \mu'$ if $\mu$ is absolutely continuous with respect to $\mu'$, and $\mu \sim \mu'$ if they are \emph{equivalent} (i.e., mutually absolutely continuous). For measures $\mu$ and $\mu'$ on $(X,\X)$ and $(Y,\Y)$, respectively, we let $\mu \otimes \mu'$ denote the product measure on $(X \times Y,\X \otimes \Y)$, where $\X\otimes \Y$ denotes the product $\sigma$-algebra.

A map $\rho: X \times \Y \rightarrow [0,1]$ is a \emph{probability kernel} if for each $B \in \Y$, the map $x \mapsto \rho(x,B)$ is $\X$-measurable and for each $x \in X$, the map $B \mapsto \rho(x,B)$ is a probability measure on $(Y,\Y)$ \citep[see, e.g.,][pp.~55--56]{Kallenberg:21:Foundations}. Given a probability measure $\mu$ on $(X,\X)$, we define $\mu \otimes \rho$ to be the probability measure on $(X \times Y,\X \otimes \Y)$ such that for each $B_1\in \X$ and $B_2\in \Y$,
\begin{align*}
	(\mu \otimes \rho)(B_1\times B_2) = \int_{B_1} \rho(x,B_2) \mu(dx).
\end{align*}
Conversely, given a probability measure $\lambda$ on $(X \times Y, \X \otimes \Y)$ with $X$-marginal $\mu$, a probability kernel $\rho: X \times \Y \to [0,1]$ is a \emph{disintegration} of $\lambda$ with respect to $\mu$ if $\lambda = \mu \otimes \rho$ \citep[see, e.g.,][p.~60]{Kallenberg:21:Foundations}. For probability kernels $\rho,\rho': X \times \Y \rightarrow [0,1]$ and a measure $\mu$ on $(X,\X)$, we say that $\rho$ is a \emph{$\mu$-version} of $\rho'$ if $\rho(x,\csdot) = \rho'(x,\csdot)$ for $\mu$-almost all $x\in X$.

\section{Incomplete Information Game}
\label{sec2}

\subsection{Types and Beliefs}

Consider an incomplete information game $G$ with $n\ge 2$ players indexed by $i=1,\ldots,n$. Let $T_0$ denote the set of payoff states that are not observed by any player, and let $T_i$, $i=1,\ldots,n$, denote the type set of player $i$. The type $t_i \in T_i$ may include payoff-relevant characteristics as well as belief-relevant information. Each player $i$ chooses an action from a non-empty topological space $A_i$. We assume that each $T_j$, $j=0,\ldots,n$, is a non-empty set equipped with a $\sigma$-algebra $\T_j$.

% Note: we cannot define a probability measure on (∅, {∅}).

\begin{assumption}[Types and Actions]
	\label{assump:basic}
	For each $i=1,\ldots,n$, $A_i$ is a compact metrizable space, and either
	(a)~$T_i$ is a Borel space with $\T_i=\B(T_i)$, or
	(b)~$\T_i$ is universally complete.\footnote{
		This assumption provides part of the sufficient conditions for the existence of measurable maximizers required for our existence result. Note that the applicable case, (a) or (b), may differ across the players.
	}
\end{assumption}

We define $A \eqdef \prod_{i=1}^n A_i$ and $T \eqdef \prod_{j=0}^n T_j$ with the corresponding $\sigma$-algebra $\T \eqdef \bigotimes_{j=0}^n \T_j$. A \emph{behavioral strategy} for player $i=1,\ldots,n$ is a probability kernel $\gamma_i: T_i \times \B(A_i) \rightarrow [0,1]$. The set of all behavioral strategies $\gamma_i$ is denoted by $\RC_i$, and the set of all such profiles by $\RC \eqdef \prod_{i=1}^n \RC_i$. We represent the belief of player $i$ as a probability kernel $\eta_i: T_i \times \T_{-i} \rightarrow [0,1]$, where $\T_{-i} \eqdef \bigotimes_{j=0, j \ne i}^n \T_j$. For each $t_i \in T_i$, the probability measure $\eta_i(t_i,\csdot)$ expresses player $i$'s belief about the payoff state and the other players' types $t_{-i} \in T_{-i}\eqdef \prod_{j=0,j \ne i}^n T_j$, when the player's type is $t_i$. We call $\eta_i$ the \emph{belief map} of player $i$ following \cite{Friedenberg/Meier:17:ET} and refer to $\eta \eqdef (\eta_1,\ldots,\eta_n)$ as the corresponding \emph{belief map profile}. For $j=0,\ldots,n$, let $\pi_j: T \rightarrow T_j$ denote the coordinate projection given by $\pi_j(t)=t_j$. The following definitions introduce the notions of consistency and absolute continuity of beliefs.

\begin{definition}[Consistency of Beliefs]
	%\label{def:consistency}
	Let $p$ be a probability measure on $(T,\T)$. A belief map profile $\eta$ is \emph{consistent under $p$} if for each $i=1,\ldots,n$, $\eta_i$ is a disintegration of $p$ with respect to the marginal $p \circ \pi_i^{-1}$ on $(T_i, \T_i)$, where $T$ is identified with $T_i \times T_{-i}$. We say that $\eta$ is \emph{consistent} if it is consistent under some probability measure on $(T,\T)$.
\end{definition}

Consistency of beliefs means that there exists a probability measure $p$ on $(T,\T)$, which can be regarded as a common prior, having the belief maps as posteriors, i.e., each $\eta_i$ is a regular conditional $p$-probability on $(T_{-i},\T_{-i})$ given player $i$'s type.

We let $\Sigma(T)$ denote the set of product probability measures on $(T,\T)$ of the form $\nu = \bigotimes_{j=0}^n \nu_j$, where $\nu_j$ is a probability measure on $(T_j,\T_j)$. Given $\nu \in \Sigma(T)$, we let $\nu_{-i}\eqdef\bigotimes_{j=0,j \ne i}^n \nu_j$, and, for $\gamma,\gamma'\in \RC$ and belief map profiles $\eta$ and $\eta'$, write $\gamma\sim_{\nu}\gamma'$ (resp.\ $\eta\sim_{\nu}\eta'$) if $\gamma_i$ is a $\nu_i$-version of $\gamma'_i$ (resp.\ $\eta_i$ is a $\nu_i$-version of $\eta'_i$) for each $i=1,\ldots, n$.

\begin{definition}[Absolute Continuity of Beliefs]
	%\label{def:abs_cont}
	Let $\nu\in \Sigma(T)$. A belief map profile $\eta$ is \emph{absolutely continuous with respect to} (or \emph{dominated by}) $\nu$ if for each $i=1,\ldots,n$ and $t_i \in T_i$,
	\[
		\eta_i(t_i,\csdot) \ll \nu_{-i}.
	\]
	In this case, we write $\eta \ll \nu$ and say that $\eta$ satisfies \emph{absolute continuity of beliefs}.
\end{definition}

Absolute continuity of beliefs requires the existence of a product probability measure $\nu$ such that each player $i$'s belief, regardless of the player's type, assigns zero probability to $\nu_{-i}$-null sets. The product structure of $\nu_{-i}$ implies that each marginal $\nu_j$ determines sets in coordinate $j$ to which all players whose beliefs concern coordinate $j$ assign probability zero, regardless of their own type. This is a new concept, distinct from the notion of absolute continuity of information introduced by \cite{Milgrom/Weber:85:MOR}.

\begin{definition}[Absolute Continuity of Information]
	A probability measure $p$ on $(T,\T)$ satisfies \emph{absolute continuity of information} if it is dominated by the product of its coordinate marginals, i.e., $p\ll \bigotimes_{j=0}^n p\circ \pi_{j}^{-1}$.
\end{definition}

Absolute continuity of information is widely used in the literature on the existence of Bayesian Nash equilibria (e.g., \citeay{Balder:88:MOR}, \citeay{Jackson/Swinkels:05:Eca}, and \citeay{He/Yannelis:16:JET}). By Lemma \ref{lemm:aci_equivalence} in the appendix, this condition on a common prior $p$ is equivalent to $p\ll \nu$ for some $\nu\in \Sigma(T)$. Consequently, absolute continuity of beliefs can be viewed as an extension of this notion to a setting without a common prior.

\begin{proposition}
	\label{prop:abs_cont}
	Suppose that a belief map profile $\eta$ is consistent under a probability measure $p$ on $(T,\T)$, and let $q\eqdef \bigotimes_{j=0}^n p\circ \pi_j^{-1}$.
	\begin{enumerate}[label=\tn{(\roman*)}]
		\item If $\eta \ll \nu$ for some $\nu \in \Sigma(T)$, then $p \ll q$.

		\item Suppose further that for each $j=0,\ldots, n$, either \tn{(a)} $T_j$ is a Borel space with $\T_j=\B(T_j)$, or \tn{(b)} $\T_j$ is the universal completion of a countably generated $\sigma$-algebra on $T_j$. If $p \ll q$, then there exists a belief map profile $\widetilde{\eta}$ such that $\widetilde{\eta} \ll q$ and $\widetilde{\eta}\sim_{q}\eta$.
	\end{enumerate}
\end{proposition}

The proof of Proposition \ref{prop:abs_cont} is relegated to Appendix \ref{app1}. The key observation for part (i) is that if no player assigns positive probability to sets that are null under the dominating product measure, and beliefs are consistent under a common prior, then the common prior itself cannot assign positive probability to such sets. For part (ii), we construct the modified belief maps from the Radon--Nikodym derivative of the common prior with respect to the dominating product measure. For each player, we obtain versions of conditional distributions that satisfy the absolute continuity condition. Consistency ensures that these conditional distributions agree with the original belief maps almost everywhere under the marginals of the common prior.

Proposition \ref{prop:abs_cont} implies that, for a belief map profile $\eta$ consistent under a common prior $p$, absolute continuity of beliefs with respect to some $\nu \in \Sigma(T)$ implies absolute continuity of information for $p$. Conversely, if the $\sigma$-algebras $\T_j$, $j=0,\ldots,n$, additionally satisfy the conditions of part (ii), then absolute continuity of information for $p$ ensures that $\eta$ admits a $q$-version $\widetilde{\eta}$ satisfying absolute continuity of beliefs with respect to $q$.

There are many examples of inconsistent belief map profiles that satisfy absolute continuity of beliefs. For instance, when $T$ is finite, every belief map profile is dominated by the uniform distribution on $T$. Examples of inconsistent beliefs with finite type sets are found in \cite{Maschler/Solan/Zamir:20:GameTheory}. The following example with uncountable type sets illustrates that absolute continuity of beliefs is substantially weaker than absolute continuity of information.

\begin{example}[Inconsistent, Absolutely Continuous Beliefs]
	\label{example:inconsistent_beliefs}
	Consider a game with two players, $1$ and $2$, whose type sets are given by $T_i = \R$, $i=1,2$. Let $T_0=\R$. For player $i$ and each $t_i \in T_i$, the belief $\eta_i(t_i,\csdot)$ is given by the normal distribution
	\[
		\N\left(
		\begin{bmatrix}
				0 \\
				r t_i
			\end{bmatrix},
		\begin{bmatrix}
				1 & 0 \\
				0 & \sigma_i^2
			\end{bmatrix}
		\right),
	\]
	where $\sigma_i^2>0$ and $r \in (0,1)$. Then, consistency of beliefs requires $\sigma_1^2 = \sigma_2^2$. To see this, denote the Gaussian densities of the belief map of player $i$ by $f_i(\csdot\mid t_i)$. If $\eta=(\eta_1,\eta_2)$ is consistent, there exist smooth and strictly positive densities $\phi_1$ and $\phi_2$ on $\R$ such that for each $t \in \R^3$,
	\begin{align*}
		f_1(t_{-1} \mid t_1) \phi_1(t_1) = f_2(t_{-2} \mid t_2) \phi_2(t_2),
	\end{align*}
	or, after taking logs,
	\begin{align}
		\label{eq:log_densities}
		\log \phi_1(t_1) - \log \phi_2(t_2) = \log f_2(t_{-2} \mid t_2) - \log f_1(t_{-1} \mid t_1).
	\end{align}
	The left-hand side of \eqref{eq:log_densities} is additively separable in $t_1$ and $t_2$, which implies that
	\begin{align*}
		0 = \frac{\partial^2 (\log f_2(t_{-2} \mid t_2) - \log f_1(t_{-1} \mid t_1))}{\partial t_1 \partial t_2} = \frac{r}{\sigma_2^2} - \frac{r}{\sigma_1^2}.
	\end{align*}
	Thus, consistency of beliefs fails when $\sigma_1^2\ne \sigma_2^2$. Conversely, if $\sigma_1^2=\sigma_2^2=\sigma^2$, $\eta$ is consistent under the common prior
	\[
		\N(0,1)\otimes \N\left(0, \frac{\sigma^2}{1-r^2}\begin{bmatrix} 1 & r \\ r & 1\end{bmatrix}\right)
	\]
	on $T_0\times T_1\times T_2$. In either case, the belief map profile is absolutely continuous with respect to the standard normal distribution on $\R^3$.
\end{example}

Nevertheless, absolute continuity of beliefs is far from innocuous. Absolute continuity of information excludes a setting where at least two players observe the same continuously distributed signal under a common prior. Similarly, absolute continuity of beliefs fails when at least two players observe the same signal ranging over a continuum, as shown in the following example inspired by an example in \cite{Cotter:91:JET}. (See \citeay{Stinchcombe:10:JET} for an extensive study of information structures beyond absolute continuity of information.)

\begin{example}[Failure of Absolute Continuity of Beliefs]
	\label{example:failure_acb}
	Suppose that $T_0 = [0,1]$ and $T_i = W \times T'_i$, $i=1,\ldots,n$, where $W = [0,1]$ and $T'_i=\R$. Player $i$'s type takes the form $t_i = (w,t_i')$, where $w$ is commonly observed by all the players. In this case, for each $j \ne i$, $j \ne 0$, the marginal of $\eta_i(t_i, \csdot)$ on the $W$-component of $T_j$ is the Dirac measure at $w$, $\delta_{w}$. As $w$ ranges over $[0,1]$, these Dirac marginals are mutually singular, so absolute continuity of beliefs fails. Indeed, a dominating $\nu \in \Sigma(T)$ would require the marginal of $\nu_{-i}$ on the $W$-component of $T_j$ with $j\ne i$ and $j\ne 0$ to assign positive mass to every $w \in [0,1]$, which is impossible for any probability measure on this set.
\end{example}

\subsection{Payoffs and Bayesian Equilibria}

The payoff function of player $i$ is given by a $(\B(A) \otimes \T)$-measurable function $u_i: A \times T \rightarrow \R$. To ensure well-behaved expected payoffs, we introduce the following assumption.
\begin{assumption}
	\label{assump:payoff-1}
	For each $i=1,\ldots,n$, there exists a $\T_i$-measurable map $\overline{u}_i: T_i \rightarrow [0,\infty)$ such that for all $a \in A$ and $t \in T$, $|u_i(a,t)| \le \overline{u}_i(t_i)$.
\end{assumption}
This assumption allows for unbounded payoffs while maintaining tractable integrability conditions through a type-dependent bound. In particular, under absolute continuity of beliefs, a dominating product measure can always be chosen so that each $\overline{u}_i$ is integrable.

\begin{lemma}
	\label{lemm:integrable_dominating}
	Suppose that Assumption \ref{assump:payoff-1} holds and $\eta\ll\nu$ for some $\nu\in\Sigma(T)$. Then, there exists $\nu'\in\Sigma(T)$ such that $\eta\ll\nu'$ and $\int\overline{u}_i(t_i)\nu_i'(dt_i)<\infty$ for all $i=1,\ldots,n$.
\end{lemma}

\begin{proof}
	Define $\nu'\in\Sigma(T)$ by setting $\nu_0'=\nu_0$ and $\nu_i'(dt_i)=c_i(1+\overline{u}_i(t_i))^{-1}\nu_i(dt_i)$ for $i=1,\ldots,n$, where $c_i$ is the normalizing constant. Then $\eta\ll\nu'$ as $\nu'\sim\nu$, and $\int\overline{u}_i(t_i)\nu_i'(dt_i)\le c_i<\infty$.
\end{proof}

For each $\gamma\in\RC$ and $t \in T$, we define
\begin{align}
	\label{eq:t_expected_payoff}
	V_i(\gamma;t) \eqdef \int_{A_1} \cdots\int_{A_n} u_i(a,t) \gamma_1(t_1,da_1) \cdots \gamma_n(t_n,da_n),
\end{align}
where $\gamma = (\gamma_1,\ldots,\gamma_n)$ and $a = (a_1,\ldots,a_n)$.\footnote{
	For fixed $t\in T$, Assumption \ref{assump:payoff-1} ensures that the multiple integral in \eqref{eq:t_expected_payoff} is finite and equals the integral with respect to the product measure $\bigotimes_{i=1}^n \gamma_i(t_i,\csdot)$ by Fubini's theorem.
}
The expected payoff of player $i$ with type $t_i$ from a profile of behavioral strategies $\gamma $ is given by
\begin{align*}
	U_i(\gamma; t_i) \eqdef \int V_i(\gamma;t) \eta_i(t_i,d t_{-i}), \quad t_i \in T_i.
\end{align*}
We write $U_i(\gamma_i,\gamma_{-i};t_i)$ and $u_i(a_i,a_{-i},t)$ with $\gamma_{-i} \eqdef (\gamma_j)_{j \ne i}$ and $a_{-i} \eqdef (a_j)_{j \ne i}\in A_{-i}\eqdef\prod_{j\ne i}A_j$, abusing notation.\footnote{
	The literature often refers to the expected payoff $U_i$ as ``\emph{interim} expected payoff'' as opposed to ``\emph{ex ante} expected payoff''. We avoid using the former term because we take the belief map profile as a primitive, and there is no notion of an \emph{ex ante} stage in the game $G$.
}
This completes the description of the game $G$. We define a Bayesian equilibrium as follows (\citeay{Simon:03:IJM}, \citeay{VanZandt:10:JET}, and \citeay{Maschler/Solan/Zamir:20:GameTheory}).

\begin{definition}
	A profile of behavioral strategies $\gamma^*\in\RC$ is a \emph{Bayesian equilibrium} of the incomplete information game $G$ if for each $i=1,\ldots,n$, $t_i \in T_i$, and $\gamma_i \in \RC_i$,\footnote{
			Equivalently, $\gamma^*$ is a Bayesian equilibrium if and only if for each $i=1,\ldots,n$ and $t_i \in T_i$, $\gamma_i^*(t_i,\csdot)$ maximizes player $i$'s expected payoff among all Borel probability measures on $A_i$. This is the best-response formulation of \cite{VanZandt:10:JET}, adapted to behavioral strategies.
	}
	\begin{align}
		\label{eq:BE}
		U_i(\gamma_i^*,\gamma_{-i}^*; t_i) \geq U_i(\gamma_i,\gamma_{-i}^*; t_i).
	\end{align}
\end{definition}

The solution concept of a Bayesian equilibrium here is stronger than that often used in settings with a common prior, where the equilibrium condition \eqref{eq:BE} is required to hold only for almost every $t_i$ (e.g., \citeay{Kim/Yannelis:97:JET}). However, without a common prior to define exceptional null sets, a Bayesian equilibrium requires \eqref{eq:BE} to hold \emph{for every} $t_i \in T_i$. In particular, a version of a Bayesian equilibrium need not itself be a Bayesian equilibrium. As a result, a typical argument that claims equivalence between distributional strategy equilibria and behavioral strategy equilibria breaks down (see \citeay{VanZandt:10:JET} for a related discussion). This also makes proving the existence of Bayesian equilibria with general type sets non-trivial.

As mentioned previously, the literature on equilibrium existence in incomplete information games often invokes consistency of beliefs under a common prior. However, we can easily find examples where Bayesian equilibria exist, even though the beliefs are inconsistent. We provide an example with uncountable type sets.

\begin{example}[Public Good Provision]
	%\label{example:public_good_provision}
	Consider a setting with two players $i=1,2$. Each player $i$ has a type set $T_i = [t_L,t_U]$, with $0 = t_L < 1 < t_U < \infty$, and chooses an action from $\{0,1\}$. Player $i$ has a payoff function given by
	\begin{align*}
		u_i(a,t) = \max\{a_1,a_2\} - a_i t_i.
	\end{align*}
	We assume that $T_0$ is a singleton. The belief $\eta'_i(t_i,\csdot)\eqdef\eta_i(t_i,T_0\times \csdot)$ on $[t_L,t_U]$ is given by a truncated normal distribution as follows:
	\begin{align*}
		\eta'_i(t_i, (t_L, t']) = \frac{\Phi((t'-t_i)/\sigma_i) - \Phi((t_L-t_i)/\sigma_i)}{\Phi((t_U-t_i)/\sigma_i) - \Phi((t_L-t_i)/\sigma_i)}, \quad t' \in [t_L,t_U],
	\end{align*}
	with $\Phi$ denoting the CDF of the standard normal distribution, and $\sigma_i^2>1/4$. These belief maps cannot be consistent if $\sigma_1^2 \ne \sigma_2^2$ by the same arguments as in Example \ref{example:inconsistent_beliefs}.

	We consider pure strategies of the form
	\begin{align*}
		\gamma_i^*(t_i,\csdot) = \delta_{1\{t_i \le t_i^*\}}(\csdot), \quad t_i^* \in [t_L,t_U],
	\end{align*}
	where $\delta_x$ is the Dirac measure at $x \in \R$. Let $p_i(t_i) \eqdef \gamma_i(t_i,\{1\})$. Then, the expected payoff of player $i$ when the other player plays $\gamma_{-i}^*$ with the cutoff $t_{-i}^{*}\in [t_L,t_U]$ is given by
	\begin{align*}
		U_i(\gamma_i,\gamma_{-i}^*;t_i) &= \int_{t_L}^{t_{-i}^*} (1 - p_i(t_i) \cdot t_i) \eta'_i(t_i,dt'_{-i}) + \int_{t_{-i}^*}^{t_U} p_i(t_i)(1 - t_i) \eta'_i(t_i,dt'_{-i}) \\
		&=1 - (1-p_i(t_i))\eta'_i(t_i,(t_{-i}^*, t_U])-p_i(t_i)\cdot t_i.
	\end{align*}
	For each $t_i\in [t_L,t_U]$, we have
	\[
		\frac{\partial U_i(\gamma_i,\gamma_{-i}^*;t_i)}{\partial p_i(t_i)}=\eta'_i(t_i,(t_{-i}^*, t_U]) - t_i.
	\]
	Therefore, player $i$ contributes with probability $1$ when $\eta'_i(t_i,(t_{-i}^*, t_U]) \ge t_i$, and does not contribute otherwise (selecting contribution at the indifference point by convention). The bound $1/4$ on $\sigma_i^2$ guarantees that $0\le \partial \eta_i'(t_i, (\csdot,t_U])/\partial t_i <1$, uniformly on $[t_L,t_U]^2$ (see Lemma \ref{lemm:trunc_normal_drv} in the appendix). Hence, $t_i\mapsto \eta'_i(t_i,(t_{-i}^*, t_U]) - t_i$ is strictly decreasing, so the pointwise best response is a cutoff strategy. Let $\Gamma: [t_L,t_U]^2 \to [t_L,t_U]^2$ be defined by
	\[
		\Gamma(t_1,t_2)=\begin{bmatrix}
			\eta'_1(t_1,(t_2,t_U]) \\
			\eta'_2(t_2,(t_1,t_U])
		\end{bmatrix}.
	\]
	Since the map $t\mapsto\Gamma(t)$ is continuous and $[t_L,t_U]^2$ is convex and compact, Brouwer's fixed point theorem ensures the existence of a fixed point $t^* = (t_1^*,t_2^*)$ such that $t^* = \Gamma(t^*)$. The corresponding cutoff strategies form a Bayesian equilibrium.
\end{example}

\section{Existence of Bayesian Equilibria}
\label{sec3}

Our main focus is on the existence of Bayesian equilibria in the game $G$. A standard existence proof invokes a fixed point theorem, which requires showing the continuity of the correspondence that maps a profile of behavioral strategies to the best responses. This continuity is usually guaranteed by the narrow topology on the space of behavioral strategies when the correspondence is in an ex ante form (\citeay{Balder:88:MOR}). Since we do not assume a common prior, such an ex ante form is not available.

We take a different approach. First, we construct a surrogate game from $G$, for which the existence of a
Nash equilibrium follows from standard results in the literature. Then, we show
that every such equilibrium has a version that is a Bayesian equilibrium of $G$.

Fix the payoff profile $u \eqdef (u_i)_{i=1}^n$, the associated dominating maps $(\overline{u}_i)_{i=1}^n$ satisfying Assumption \ref{assump:payoff-1}, and the belief map profile $\eta$ of the game $G$. Define $\M$ as the subset of the product probability measures in $\Sigma(T)$ that dominate the belief map profile $\eta$ and for which $\overline{u}_i$ is integrable with respect to $\nu_i$ for all $i=1,\ldots,n$, i.e.,
\begin{align*}
	\M \eqdef \left\{\nu \in \Sigma(T) : \eta \ll \nu \text{ and } \max_{1\le i\le n}\int \overline{u}_i(t_i)\, \nu_i(dt_i) < \infty \right\}.
\end{align*}
By Lemma \ref{lemm:integrable_dominating}, $\M\ne\varnothing$ if and only if $\eta$ satisfies absolute continuity of beliefs, regardless of the choice of the dominating maps. For $\nu \in \M$, we construct a surrogate game, denoted by $G^*(\nu)$, by taking $\RC_i$ as player $i$'s action space and defining player $i$'s payoff at $\gamma \in \RC$ by
\begin{align}
	\label{eq:aux_payoff}
		\widetilde{U}_i(\gamma;\nu) \eqdef\int U_i(\gamma;t_i) \nu_i(dt_i).
\end{align}

\begin{definition}
	A profile of actions $\widetilde{\gamma}^* \in \RC$ is a \emph{Nash equilibrium} of a game $G^*(\nu)$, with $\nu\in \M$, if for each $i=1,\ldots,n$ and $\gamma_i \in \RC_i$,
	\begin{align}
		\label{eq:NE}
		\widetilde{U}_i(\widetilde{\gamma}_i^*,\widetilde{\gamma}_{-i}^*;\nu) \ge \widetilde{U}_i(\gamma_i,\widetilde{\gamma}_{-i}^*;\nu).
	\end{align}
\end{definition}

We now relate the Bayesian equilibria of the original game $G$ to the Nash equilibria of the surrogate games $G^*(\nu)$, $\nu \in \M$. Let $\BE$ denote the set of Bayesian equilibria in the game $G$. For each $\nu \in \M$, define $\NE_\nu$ to be the set of Nash equilibria in the game $G^*(\nu)$. Also, let
\[
	\NE \eqdef \bigcap_{\nu \in \M} \NE_\nu,
\]
the set of $\gamma \in \RC$ that are Nash equilibria in $G^*(\nu)$ for all $\nu \in \M$. We present a characterization of $\BE$ as follows.

\begin{theorem}[Characterization of Bayesian Equilibria]
	\label{thm:char}
	Suppose that Assumption \ref{assump:payoff-1} holds and $\eta$ satisfies absolute continuity of beliefs. Then,
	\[
		\BE = \NE.
	\]
\end{theorem}

\begin{proof}
	Suppose that $\gamma^* \in \BE$. Using \eqref{eq:aux_payoff}, for any $i=1,\ldots, n$, $\gamma_i \in \RC_i$, and $\nu\in \M$, we have
	\begin{align*}
		\widetilde{U}_i(\gamma_i^*,\gamma_{-i}^*;\nu) \ge \widetilde{U}_i(\gamma_i,\gamma_{-i}^*;\nu).
	\end{align*}
	Hence, $\gamma^* \in \NE_\nu$ for all $\nu \in \M$, i.e., $\BE \subset \NE$.

	To show the other inclusion, take $\gamma^* \in\NE$. Suppose that $\gamma^* \notin \BE$, i.e., there exist $i \in \{1,\ldots,n\}$, $t'_i \in T_i$, and $\gamma_i \in \RC_i$ such that
	\begin{align*}
		U_i(\gamma_i^*,\gamma_{-i}^*;t'_i) < U_i(\gamma_i,\gamma_{-i}^*;t'_i).
	\end{align*}
	By Lemma \ref{lemm:integrable_dominating}, $\M\ne\varnothing$. Choose any $\nu \in \M$ and consider $\nu' \in \Sigma(T)$ such that $\nu_i' = \lambda \nu_i + (1-\lambda) \delta_{t'_i}$ and $\nu_{-i}' = \nu_{-i}$, where $\lambda \in (0,1]$ and $\delta_x$ is the Dirac measure at $x\in T_i$. Then $\nu_i \ll \nu_i'$ and, by Assumption \ref{assump:payoff-1},
	\begin{align*}
		\int \overline{u}_i(t_i) \nu_i'(dt_i) = \lambda \int \overline{u}_i(t_i) \nu_i(dt_i) + (1-\lambda) \overline{u}_i(t'_i) < \infty.
	\end{align*}
	Hence, $\nu' \in \M$. In light of \eqref{eq:aux_payoff}, we can write
	\begin{align}
		\label{eq:payoff_diff}
		\begin{aligned}
			\widetilde{U}_i(\gamma_i^*,\gamma_{-i}^*; \nu') - \widetilde{U}_i(\gamma_i,\gamma_{-i}^*; \nu') &= \lambda \left(\widetilde{U}_i(\gamma_i^*,\gamma_{-i}^*; \nu) - \widetilde{U}_i(\gamma_i,\gamma_{-i}^*; \nu)\right) \\
			&\quad + (1-\lambda) \left(U_i(\gamma_i^*,\gamma_{-i}^*; t'_i) - U_i(\gamma_i,\gamma_{-i}^*; t'_i)\right).
		\end{aligned}
	\end{align}
	Since the first difference on the right-hand side is finite by the integrability condition in the definition of $\M$, taking $\lambda$ sufficiently small makes the difference in the expected payoffs in \eqref{eq:payoff_diff} negative, so that $\gamma_i$ is a profitable deviation for player $i$ in the game $G^*(\nu')$. Consequently, $\gamma^* \notin \NE_{\nu'}$, which contradicts $\gamma^*\in \NE$.
\end{proof}

\begin{remark*}
	Theorem \ref{thm:char} is useful for analyzing the set of Bayesian equilibria in the game $G$ through the Nash equilibria of the surrogate games $G^*(\nu)$, $\nu \in \M$. One can construct an alternative characterization of $\BE$ in a rather trivial manner. First, take $\M_{\delta} \subset \Sigma(T)$, where each $T_i$-marginal of $\nu\in \M_{\delta}$, $i=1,\ldots,n$, is a Dirac measure. (For such $\nu$, the surrogate game $G^*(\nu)$ and its set of Nash equilibria $\NE_\nu$ are defined by \eqref{eq:aux_payoff} exactly as for $\nu\in\M$.) Then, we trivially have
	\begin{align*}
		\BE = \bigcap_{\nu \in \M_{\delta}} \NE_\nu.
	\end{align*}
	Note that this characterization neither implies nor follows from our characterization result because the sets $\M_{\delta}$ and $\M$ are generally incomparable (neither contains the other). This characterization is not very useful because an action profile $\gamma$ that is a Nash equilibrium in the game $G^*(\nu)$ \emph{for all} $\nu \in \M_\delta$ is not guaranteed to exist. In contrast, the characterization in Theorem \ref{thm:char} involves $G^*(\nu)$ \emph{only for} $\nu \in \M$. For each such $\nu$, standard existence results apply to $G^*(\nu)$ because, as shown below, it can be transformed into a Bayesian game whose common prior $\nu$ is a product measure and therefore satisfies absolute continuity of information.
\end{remark*}

\subsection{Main Results}

The characterization theorem does not directly lead to the existence of Bayesian equilibria in the game $G$ because it involves all the product probability measures in $\M$. To address this, we introduce an additional assumption on the payoff functions.

\begin{assumption}
	\label{assump:payoff-2}
	For each $i=1,\ldots,n$, $a_{-i} \in A_{-i}$, and $t \in T$, the map $a_i \mapsto u_i(a_i,a_{-i},t)$ is upper semicontinuous on $A_i$.\footnote{
		This condition together with Assumption \ref{assump:basic} is used to guarantee the existence of measurable maximizers in the proof of our main result.
	}
\end{assumption}

Assumption \ref{assump:payoff-2} allows for payoffs that are discontinuous in actions, requiring only upper semicontinuity in a player's own action. It is weaker than the joint continuity in actions assumed in \cite{Balder:88:MOR}. Many natural games satisfy the stronger condition of joint upper semicontinuity in actions, as illustrated in the example below.

\begin{example}[Public Good Provision with a Threshold]
	%\label{ex:threshold-pg}
	Consider the threshold public good provision (contribution) game studied by \cite{Menezes/Monteiro/Temimi:01:JME}, where each player $i$ chooses a contribution $a_i \in [0, \overline a]$, and the public good is provided if and only if $\sum_{j=1}^n a_j \ge K> 0$, with contributions not refunded if the threshold is not met. The payoff function is
	\begin{align*}
		u_i(a, t) = \phi_i(t_i) \cdot 1\!\left\{\sum_{j=1}^n a_j \ge K\right\} - a_i,
	\end{align*}
	where $\phi_i(t_i) > 0$ is the type-dependent value of the good. It is straightforward to verify that the map $a\mapsto u_i(a, t)$ is upper semicontinuous for each $t\in T$.
\end{example}

We start with a lemma showing that a Nash equilibrium of a surrogate game yields a Bayesian equilibrium of the original game.

\begin{lemma}
\label{lemm:BE_main}
Suppose that Assumptions \ref{assump:basic}, \ref{assump:payoff-1}, and \ref{assump:payoff-2} hold and let $\widetilde{\gamma}^*\in \NE_\nu$ for some $\nu\in \M$. Then, there exists a Bayesian equilibrium $\gamma^*$ in the game $G$ such that $\gamma^*\sim_{\nu}\widetilde{\gamma}^*$.
\end{lemma}

The proof of Lemma \ref{lemm:BE_main} is relegated to Appendix \ref{app1}. In a nutshell, it proceeds by modifying the Nash equilibrium $\widetilde{\gamma}^*$ of the surrogate game $G^*(\nu)$ into a Bayesian equilibrium of the original game $G$. For each player $i$, we partition $T_i$ into the set $\widetilde{T}_i$ of types at which $\widetilde{\gamma}_i^*$ achieves the maximal expected payoff against $\widetilde{\gamma}_{-i}^*$, and its complement. On $\widetilde{T}_i$, we retain the original strategy, while on $\widetilde{T}_i^c$, we replace it with a measurable best-response selector, which exists under the stated assumptions. The Nash equilibrium condition \eqref{eq:NE} implies that $\widetilde{T}_i^c$ is a $\nu_i$-null set, so the modified strategy $\gamma_i^*$ is a $\nu_i$-version of $\widetilde{\gamma}_i^*$. The key observation is that modifying strategies on $\nu_i$-null sets does not affect the other players' expected payoffs, which follows from absolute continuity of beliefs. Consequently, for each player $i$ and type $t_i \in T_i$, facing the modified opponents' strategies $\gamma_{-i}^*$ yields the same expected payoff as facing $\widetilde{\gamma}_{-i}^*$, while the modified own strategy $\gamma_i^*$ achieves optimality at every type. This makes $\gamma^*$ a Bayesian equilibrium.

One consequence of Lemma \ref{lemm:BE_main} is that the prediction of a game $G^*(\nu)$ remains invariant under changing the product probability measure $\nu\in \M$ to any equivalent $\nu'\in \M$.

\begin{proposition}
	\label{prop:mutual_abs_cont}
	Suppose that Assumptions \ref{assump:basic}, \ref{assump:payoff-1}, and \ref{assump:payoff-2} hold. Suppose further that $\nu,\nu' \in \M$ are such that $\nu \sim \nu'$. Then,
	\begin{align*}
		\NE_\nu = \NE_{\nu'}.
	\end{align*}
\end{proposition}

\begin{proof}
	Suppose that $\widetilde{\gamma}^*\in \NE_\nu$. By Lemma \ref{lemm:BE_main}, there exists a Bayesian equilibrium $\gamma^*$ in the game $G$ such that $\gamma^* \sim_\nu \widetilde{\gamma}^*$. Then, by Theorem \ref{thm:char}, $\gamma^*\in\NE_{\nu'}$, and since $\nu \sim \nu'$, we have $\gamma^*\sim_{\nu'}\widetilde{\gamma}^*$, implying $\widetilde{\gamma}^*\in \NE_{\nu'}$ as $\nu'\in \M$. Thus, $\NE_\nu \subset \NE_{\nu'}$. Reversing the roles of $\nu$ and $\nu'$ gives the other inclusion, $\NE_{\nu'}\subset\NE_\nu$.
\end{proof}

In light of Proposition \ref{prop:mutual_abs_cont}, if we partition $\M$ into equivalence classes where two product probability measures $\nu$ and $\nu'$ belong to the same class if and only if $\nu \sim \nu'$, the set $\NE_\nu$ depends only on the equivalence class of $\nu$. The main reason for this result is the following. Since we do not require the belief maps to be consistent under $\nu$, we can fix these belief maps while changing $\nu$ within $\M$. This does not alter the prediction of the game as long as measure zero sets remain invariant under the change of the product probability measures.

We turn to the existence result. The following theorem shows that a Bayesian equilibrium exists in the game $G$ if and only if a Nash equilibrium exists in the surrogate game $G^*(\nu)$ for some $\nu \in \M$.

\begin{theorem}[Existence of Bayesian Equilibria]
	\label{thm:existence}
	Suppose that Assumptions \ref{assump:basic}, \ref{assump:payoff-1}, and \ref{assump:payoff-2} hold and $\eta$ satisfies absolute continuity of beliefs. Then, the following statements are equivalent.
	\begin{enumerate}[label=\tn{(\roman*)}]
		\item $\BE \ne \varnothing$.

		\item $\NE_\nu \ne \varnothing$ for some $\nu \in \M$.
	\end{enumerate}
\end{theorem}

\begin{proof}
	By Theorem \ref{thm:char}, (i) implies (ii). Conversely, Lemma \ref{lemm:BE_main} shows that (ii) implies (i).
\end{proof}

Theorem \ref{thm:existence} offers a general method for establishing the existence of a Bayesian equilibrium in a game without a common prior by applying existing results on Nash equilibria. Indeed, $G^*(\nu)$ is a normal-form game whose action spaces are the sets of behavioral strategies. The existence of a Nash equilibrium in this setting is well studied, particularly for discontinuous games (see, e.g., \citeay{Reny:99:Eca}, \citeay{Monteiro/Page:07:JET}, \citeay{McLennan/Monteiro/Tourky:11:Eca}, and \citeay{Barelli/Meneghel:13:Eca}).

The absolute continuity of beliefs requirement in Theorem \ref{thm:existence} is not innocuous. \cite{Simon:03:IJM} constructed an incomplete information game with a finite action set where there is no measurable Bayesian equilibrium, and \cite{Hellman:14:JET} presented a game with no $\epsilon$-Bayesian equilibrium. In the example below, we consider a modified version of this game, presented in \cite{Friedenberg/Meier:17:ET}, and show that the belief maps fail to satisfy absolute continuity of beliefs.

\begin{example}[\citeay{Hellman:14:JET}, \citeay{Friedenberg/Meier:17:ET}]
	\label{example:game_that_fails_abs_cont_beliefs}
	Consider the following game with two players $i=1,2$. The set of payoff states is $T_0=\{-1,1\}$, and the action space for player $i$ is given by $A_i = \{L_i,M_i\}$. For each value $t_0\in T_0$, the payoff matrix is given in Table \ref{table:payoff_matrices_FM}.

	\begin{table}[t]
		\def\arraystretch{1.2}
		\centering
		\begin{minipage}{0.45\textwidth}
			\centering
			\begin{tabular}{|x{1.4cm}|c|c|}
				\hline
				$t_0=-1$ & $L_2$& $M_2$\\
				\hline
				$L_1$ & $3$, $3$ & $2$, $2$ \\
				\hline
				$M_1$ & $2$, $2$ & $3$, $3$ \\
				\hline
			\end{tabular}
		\end{minipage}
		\begin{minipage}{0.45\textwidth}
			\centering
			\begin{tabular}{|x{1.4cm}|c|c|}
				\hline
				$t_0=1$ & $L_2$& $M_2$\\
				\hline
				$L_1$ & $4$, $4$ & $6$, $7$ \\
				\hline
				$M_1$ & $7$, $6$ & $4$, $4$ \\
				\hline
			\end{tabular}
		\end{minipage}
		\vspace{1.2em}
		\caption{\footnotesize The payoff matrices, replicated from Figure 12 of \cite{Friedenberg/Meier:17:ET}.}
		\label{table:payoff_matrices_FM}
	\end{table}

	To construct the type sets $T_i$, let $X \eqdef \{-1,1\}^{\NN \cup \{0\}}$ be the Cantor space with generic elements $x = (x_0,x_1,\ldots)$, $x_j \in \{-1,1\}$. Let $\Omega \eqdef \Omega_1 \cup \Omega_2$, where $\Omega_i \eqdef \{i\} \times X$. We endow $X$ with the Borel $\sigma$-algebra of the product topology. Define the maps $f: \Omega \rightarrow \Omega$ and $w: \Omega \rightarrow T_0$ by $f(i,x_0,x_1,\ldots)=(-i,x_1,x_2,\ldots)$, where $-i$ represents the opponent of player $i$, and $w(i,x_0,x_1,\ldots) = x_0\in T_0$, respectively. Note that
	\begin{align*}
		f^{-1}(\{(i,x_0,x_1,\ldots)\}) = \{(-i,1,x_0,x_1,\ldots),(-i,-1,x_0,x_1,\ldots)\}.
	\end{align*}
	The type set of player $i$ is given by
	\begin{align*}
		T_i = \left\{\{\omega\} \cup f^{-1}(\{\omega\}): \omega \in \Omega_i\right\}.
	\end{align*}

	To introduce belief maps, define the bijection $\tau_i: T_i \rightarrow \Omega_i$ by $\tau_i(t_i) = \omega$ if and only if $t_i = \{\omega\} \cup f^{-1}(\{\omega\})$. We endow $T_i$ with the coarsest topology that makes $\tau_i$ continuous. Since $\Omega_i$ is compact metrizable and $\tau_i$ is a homeomorphism, $T_i$ is a compact metrizable space. Finally, consider the belief map $\eta_i$ given by $\eta_i(t_i,\{(t_0,t'_{-i})\}) = 1$ if $w(\tau_i(t_i)) = t_0$ and $\tau_i(t_i)\in t'_{-i}$. \cite{Hellman:14:JET} (see also the Online Appendix of \citeay{Friedenberg/Meier:17:ET}) showed that this game does not have a universally measurable Bayesian equilibrium.

	The belief maps $\eta_i$, $i=1,2$, in this example fail the absolute continuity condition. To see this, note that for each $t_i\in T_i$, the support of $\eta_i(t_i,\csdot)$ is a singleton, say, $\{\varphi_i(t_i)\}$. By the construction of $\eta_i$, the resulting map $\varphi_i:T_i\to T_{-i}$ is injective. Since $T_i$ is uncountable, the image $\varphi_i(T_i)$ is an uncountable set. Thus, $\eta_i$ cannot be dominated by a $\sigma$-finite measure on $T_{-i}$ because no such measure can assign positive mass to each element of $\varphi_i(T_i)$.
\end{example}

The common prior assumption in the literature is often motivated by the tractability it provides for equilibrium analysis. One might worry that this tractability vanishes when the common prior assumption is relaxed. We present a result that characterizes the set of Bayesian equilibria as the set of Bayesian Nash equilibria in a surrogate incomplete information game up to null sets. We first reformulate the game $G^*(\nu)$ into a more familiar Bayesian game with a common prior as follows. Note that for each $\nu \in \M$ and $i=1,\ldots,n$, $\nu_i \otimes \eta_i \ll \nu_i \otimes \nu_{-i}$, and let $f_{\nu,i}$ be a Radon--Nikodym derivative of $\nu_i \otimes \eta_i$ with respect to $\nu_i \otimes \nu_{-i}$. For $a\in A$, $t\in T$, and $\gamma\in \RC$, define
\begin{align*}
	\widetilde{u}_i(a,t;\nu) \eqdef u_i(a,t) \cdot f_{\nu,i}(t_i,t_{-i})
\end{align*}
and
\begin{align*}
	\widetilde{V}_i(\gamma;t,\nu) \eqdef \int_{A_1} \cdots \int_{A_n} \widetilde{u}_i(a,t;\nu) \gamma_1(t_1,da_1) \cdots \gamma_n(t_n,da_n).
\end{align*}
Given $\gamma \in \RC$, we rewrite the payoff function of player $i$ in \eqref{eq:aux_payoff} as follows:
\begin{align}
	\label{eq:aux_payoff2}
	\widetilde{U}_i(\gamma;\nu) = \int \widetilde{V}_i(\gamma;t,\nu) \nu(d t).
\end{align}

Under this reformulation, the surrogate Bayesian game consists of $n$ players, each with payoff function $\widetilde{u}_i$ and private type $t_i \in T_i$, where $t\in T$ is drawn from the product probability measure $\nu$. Player $i$'s strategy space is the set of behavioral strategies $\RC_i$. Then, the ex ante expected payoff of player $i$ when the players play the behavioral strategies $\gamma_1,\ldots,\gamma_n$ is given by \eqref{eq:aux_payoff2}. We denote this incomplete information game by $\widetilde{G}^*(\nu)$. Note that each player $i$'s belief about the payoff state and the other players' types is given by $\nu_{-i}$ in this game. Hence, the product measure $\nu$ serves as the common prior. A Bayesian Nash equilibrium of $\widetilde{G}^*(\nu)$ is defined to be a behavioral strategy profile $(\widetilde{\gamma}_1^*,\ldots,\widetilde{\gamma}_n^*)$ such that \eqref{eq:NE} holds for each $i =1,\ldots,n$ and $\gamma_i \in \RC_i$. Let the set of Bayesian Nash equilibria of this game be denoted by $\BNE_{\nu}$. By definition, we have $\NE_\nu = \BNE_\nu$. Thus, we obtain the following characterization of Bayesian equilibria up to $\nu$-null sets.

\begin{proposition}
	%\label{prop:char}
	Suppose that Assumptions \ref{assump:basic}, \ref{assump:payoff-1}, and \ref{assump:payoff-2} hold and $\eta$ satisfies absolute continuity of beliefs. Then, for each $\nu \in \M$,
	\begin{align*}
		\BNE_\nu = \overline{\BE}_{\nu},
	\end{align*}
	where $\overline{\BE}_{\nu} \eqdef \left\{\gamma \in \RC: \gamma \sim_\nu \gamma'\text{ for some } \gamma' \in \BE\right\}$.
\end{proposition}

\begin{proof}
	Note that for each $\nu \in \M$, $\NE_\nu = \BNE_\nu.$ Hence, by Lemma \ref{lemm:BE_main}, $\BNE_\nu \subset \overline{\BE}_\nu$ for all $\nu \in \M$. On the other hand, fix $\nu\in\M$, take $\gamma \in \overline{\BE}_\nu$, and let $\gamma^* \in \BE$ be such that $\gamma \sim_\nu \gamma^*$. Then, $\gamma^* \in \NE_\nu$ by Theorem \ref{thm:char}, implying $\gamma \in \NE_\nu$.
\end{proof}

An immediate consequence of this proposition (or of Theorem \ref{thm:existence}) is that $\BE \ne \varnothing$ if and only if $\BNE_\nu \ne \varnothing$ for some $\nu \in \M$. Combining this result with the existence result of \cite{Balder:88:MOR}, we can show that whenever the payoff functions are continuous in actions and the other conditions of Theorem \ref{thm:existence} are satisfied, a Bayesian equilibrium exists for the game $G$. As this can be useful in many applications, we state this result as a corollary below.

\begin{corollary}
	\label{cor:BE_existence_1}
	Suppose that Assumptions \ref{assump:basic}--\ref{assump:payoff-1} hold and $\eta$ satisfies absolute continuity of beliefs. Furthermore, suppose that for each $i=1,\ldots,n$ and $t \in T$, the map $a\mapsto u_i(a,t)$ is continuous on $A$. Then, there exists a Bayesian equilibrium in the game $G$.
\end{corollary}

The proof of Corollary \ref{cor:BE_existence_1} is relegated to Appendix \ref{app1}. It proceeds by constructing a surrogate Bayesian game having a product measure as a common prior, and verifying the conditions of Theorem~3.1 of \cite{Balder:88:MOR}. The continuity of $u_i$ in actions is stronger than Assumption \ref{assump:payoff-2} and ensures that these conditions are satisfied.

As mentioned before, our existence result extends to games with payoffs discontinuous in actions under certain conditions on the payoff functions (see \citeay{Reny:20:ARE} for a review of the literature on equilibrium existence in discontinuous games). While we believe that the result applies to many such settings, for brevity and concreteness, we will draw on the result from \cite{Carbonell-Nicolau/McLean:18:MOR} as an illustration.

Let $\H_i$ be the set of pure strategies for player $i$, i.e., $\H_i$ is the set of measurable maps from $(T_i,\T_i)$ to $(A_i,\B(A_i))$. We adopt the following definition introduced by \cite{Carbonell-Nicolau/McLean:18:MOR}, who adapted concepts from \cite{Monteiro/Page:07:JET} to incomplete information games.

\begin{definition}%[\citeay{Carbonell-Nicolau/McLean:18:MOR}]
	%\label{def:uniform_payoff_security}
	A payoff profile $u$ is \emph{uniformly payoff secure} if for each $i =1,\ldots,n$, $\epsilon > 0$, and $h_i \in \H_i$, there exists $h_i^* \in \H_i$ such that for each $a_{-i}\in A_{-i}$ and $t \in T$, there is a neighborhood $N(a_{-i})$ of $a_{-i}$ with
	\begin{align*}
		u_i(h_i^*(t_i),a'_{-i},t) > u_i(h_i(t_i),a_{-i},t) - \epsilon
	\end{align*}
	for all $a'_{-i} \in N(a_{-i})$.
\end{definition}

The following result is obtained by adapting Theorem 1 of \cite{Carbonell-Nicolau/McLean:18:MOR} to our setting and is useful for establishing the existence of Bayesian equilibria in the game $G$ with potentially discontinuous payoff functions. The main distinction from their setting is that we do not require the common prior assumption, allowing the belief map profile $\eta$ to be inconsistent.
\begin{corollary}
	\label{cor:BE_existence_2}
	Suppose that Assumptions \ref{assump:basic}--\ref{assump:payoff-1} hold and $\eta$ satisfies absolute continuity of beliefs. Furthermore, suppose that $u$ is uniformly payoff secure, and for each $i=1,\ldots,n$ and $t \in T$, the map $a\mapsto u_i(a,t)$ is upper semicontinuous on $A$. Then, there exists a Bayesian equilibrium in the game $G$.
\end{corollary}

\begin{proof}
	The proof follows that of Theorem 1 of \cite{Carbonell-Nicolau/McLean:18:MOR}, with $u_i(\csdot,\csdot)f(\csdot)$ in their paper replaced by $u_i(a,t)f_{\nu,i}(t_i,t_{-i})$. Details are omitted.
\end{proof}

Below we examine an application from \cite{Carbonell-Nicolau/McLean:18:MOR}: Cournot games, which have been studied extensively in the literature. For a review, we refer the reader to their paper.

\begin{example}[Incomplete Information Cournot Game]
	%\label{example:Cournot}
	We consider a market for a single homogeneous good with $n$ firms competing in quantities. Let $T_1,\ldots,T_n$ denote the firms' type sets and let $T_0$ denote the set of common payoff states about which the firms are uncertain. We assume that $(T_i,\T_i)$, $i=1,\ldots, n$, satisfy Assumption \ref{assump:basic}. Each firm $i$ chooses a quantity $a_i$ from a compact set $A_i \subseteq [0,\infty)$. The payoff function of firm $i$ is given by
	\begin{align*}
		%\label{eq:payoff_Cournot}
		u_i(a,t) = a_i p\left(\sum_{j=1}^n a_j,t \right) - c_i(a_i,t),
	\end{align*}
	where $p(q,t)$ denotes the price of the good at aggregate output $q$ and $c_i(a_i,t)$ is firm $i$'s cost of producing $a_i$, given the common payoff state $t_0$ and the firms' individual payoff types $t_1,\ldots,t_n$.

	We introduce the following conditions on the maps $p$ and $c_i$.
	\begin{enumerate}[label=(\alph*),topsep=0em]
		\item The maps $p,c_i: [0,\infty) \times T \to [0,\infty)$ are $\B([0,\infty)) \otimes \T$-measurable and bounded.

		\item For each $t \in T$, the map $q \mapsto p(q,t)$ is continuous on $\left\{\sum_{i=1}^n a_i: a_i \in A_i \right\}$.

		\item For each $i=1,\ldots,n$ and $t \in T$, the map $a_i \mapsto c_i(a_i,t)$ is lower semicontinuous.
	\end{enumerate}

	Conditions (a) and (b) are taken from \cite{Carbonell-Nicolau/McLean:18:MOR}, and (c) is a stronger variant of their assumption on the cost function, needed to satisfy the conditions of Corollary \ref{cor:BE_existence_2}.\footnote{
		\cite{Carbonell-Nicolau/McLean:18:MOR} impose lower semicontinuity only on the aggregate cost $\sum_{i=1}^n c_i(\csdot,t)$. In our setting, the relevant object is instead the weighted aggregate cost $\sum_{i=1}^n c_i(\csdot,t)f_{\nu,i}(t_i,t_{-i})$, and (c) provides a primitive condition that ensures its lower semicontinuity.
	}
	This completes the description of the Cournot game, which we denote by $G_{\mathsf{C}}$. We obtain the following result by applying Corollary \ref{cor:BE_existence_2}.

	\begin{corollary}
		%\label{cor:Cournot}
		Suppose that the belief map profile $\eta$ in the game $G_{\mathsf{C}}$ satisfies absolute continuity of beliefs. Then, $G_{\mathsf{C}}$ has a Bayesian equilibrium.
	\end{corollary}

	\begin{proof}
		By following the proof of Corollary 4 of \cite{Carbonell-Nicolau/McLean:18:MOR}, we find that $u$ is uniformly payoff secure. Since for each $i=1,\ldots, n$ and $t\in T$, the map $a\mapsto u_i(a,t)$ is upper semicontinuous on $A$, the desired result follows from Corollary \ref{cor:BE_existence_2}.
	\end{proof}

	This corollary extends existence results for Cournot games to the setting with potentially inconsistent beliefs, at the cost of a stronger joint upper semicontinuity requirement on the payoffs.
\end{example}

\begin{remark*}
	Both the failure of absolute continuity of beliefs in Example \ref{example:failure_acb} and the non-existence in Example \ref{example:game_that_fails_abs_cont_beliefs} rely essentially on the relevant type sets being uncountable. When each $T_i$, $i=1,\ldots, n$, is at most countable, with $\T_i = 2^{T_i}$, the absolute continuity condition holds automatically. Indeed, $\eta\ll\nu$ for any $\nu\in\Sigma(T)$ such that each $\nu_i$, $i=1,\ldots,n$, assigns positive mass to every point in $T_i$ and $\nu_0$ is a mixture, with positive weights, of the countably many $T_0$-marginals of the players' beliefs. The existence of a Bayesian equilibrium follows directly from Corollary \ref{cor:BE_existence_1} for continuous payoffs, or Corollary \ref{cor:BE_existence_2} for discontinuous payoffs satisfying its conditions. In the special case where $T_0$, each $T_i$, and each $A_i$ are finite, this also recovers Theorem 10.42 of \cite{Maschler/Solan/Zamir:20:GameTheory}, which establishes existence in the belief-space framework via the agent-form game. Absolute continuity of beliefs is thus a substantive restriction only when some type sets are uncountable.
\end{remark*}

\medskip

\appendix
\bookmarksetupnext{level=part}
\section*{Appendix}

\renewcommand{\thesection}{\Alph{section}}
\makeatletter
\let\sectionname\@empty
% Make appendix sections appear as subsection level in bookmarks
\def\toclevel@section{1}
\makeatother

\section{Omitted Proofs}
\label{app1}

\begin{lemma}
	\label{lemm:aci_equivalence}
	Let $p$ be a probability measure on $(T,\T)$ and let $q\eqdef \bigotimes_{j=0}^n p\circ \pi_j^{-1}$. Then, $p\ll q$ if and only if $p\ll \nu$ for some $\nu\in \Sigma(T)$.
\end{lemma}

\begin{proof}
	The necessity is trivial as $q\in \Sigma(T)$. We prove sufficiency. Suppose that $p$ has a density $f$ with respect to some $\nu \in \Sigma(T)$. For each $j=0,\ldots, n$, let $p_{j}\eqdef p\circ \pi_j^{-1}$ and let $g_j$ be a version of $dp_j/d\nu_j$ taking values in $[0,\infty)$ (note that $p_j\ll \nu_{j}$). By standard arguments, $q\ll \nu$ and $g\eqdef\prod_{j=0}^n g_j=dq/d\nu$, $\nu$-a.e. We define $h\eqdef f/g \cdot 1\{g>0\}$ and claim that $h$ is a Radon--Nikodym derivative of $p$ with respect to $q$. Indeed, for any $E\in \T$,
	\[
		\int_E h\,dq = \int_E hg\,d\nu=\int_{E\cap \{g>0\}}f\, d\nu = p(E\cap \{g>0\}).
	\]
	It suffices to show that $p(\{g=0\})=0$. Noticing that $\{g=0\}= \bigcup_{j=0}^n \pi_j^{-1}(\{g_j=0\})$, we have
	\[
		p(\{g=0\})\le \sum_{j=0}^n p_j(\{g_j=0\})=\sum_{j=0}^n \int_{\{g_j=0\}} g_j\, d\nu_j=0,
	\]
	as required.
\end{proof}

\begin{proof}[Proof of Proposition \ref{prop:abs_cont}]
	For each $j=0,\ldots,n$, let $p_j\eqdef p\circ \pi_{j}^{-1}$ so that $q=\bigotimes_{j=0}^n p_j$.

	(i) By Lemma \ref{lemm:aci_equivalence}, it suffices to show that $p\ll \nu$. Consider $N\in \T$ with $\nu(N)=0$. For each $i=1,\ldots,n$ and $\nu_i$-almost all $t_i$, we have $\nu_{-i}(N^{t_i})=0$, where $N^{t_i}\eqdef\{t_{-i}\in T_{-i}: t\in N\}$ is the $t_i$-section of $N$. Since $\eta \ll \nu$, we have $\eta_i(t_i,N^{t_i})=0$ for $\nu_i$-almost all $t_i$. Furthermore, $p_i \ll \nu_i$. To see this, consider $i = 1$ without loss of generality. Note that for $N_1 \in \T_1$ with $\nu_1(N_1)=0$,
	\[
		p_1(N_1) =\int_{T_2} \eta_2(t_2, T_0 \times N_1 \times T_3 \times \cdots \times T_n) p_2(dt_2)=0
	\]
	because $\eta\ll \nu$, and therefore, $\eta_2(t_2, T_0\times N_1 \times T_3 \times \cdots \times T_n) = 0$ for all $t_2 \in T_2$. Thus, $p_1 \ll \nu_1$. Since $\eta$ is consistent under $p$, we get
	\[
		p(N)=\int_{T_i} \eta_i(t_i, N^{t_i})p_i(dt_i)=0.
	\]

	(ii) Fix $i\in \{1,\ldots, n\}$. For each $j=0,\ldots,n$, set $\T_{0,j} \eqdef \T_j$ under case (a), and let $\T_{0,j}$ denote a countably generated $\sigma$-algebra whose universal completion equals $\T_j$ under case (b). In both cases, each $\T_{0,j}$ is countably generated. Let $q_{-i}\eqdef \bigotimes_{j=0,j\ne i}^{n} p_j$ and let $f$ be a version of $dp/dq$, taking values in $[0,\infty)$. Define
	\[
		h_i(t_i)\eqdef \int_{T_{-i}} f(t) q_{-i}(d t_{-i}),
	\]
	$H_i \eqdef \{t_i\in T_i:h_i(t_i)=1\}$, and $\widetilde{\eta}_i: T_i \times \T_{-i} \to [0,1]$ by setting
	\begin{align*}
		\widetilde{\eta}_i(t_i,E) = 1_{H_i}(t_i)\int_{E} f(t) q_{-i}(d t_{-i}) + 1_{H_i^c}(t_i)q_{-i}(E), \quad E\in \T_{-i}.
	\end{align*}
	Note that $\widetilde{\eta}_i$ is a probability kernel because for each $E\in \T_{-i}$, the map $t_i\mapsto \widetilde{\eta}_i(t_i, E)$ is $\T_i$-measurable by Lemma \ref{lemm:mble_integral1} in Section \ref{app2} below, and for each $t_i\in T_i$, the map $E\mapsto\widetilde{\eta}_i(t_i, E)$ is clearly a probability measure. By construction, $\widetilde{\eta}_i(t_i,\csdot)\ll q_{-i}$ for each $t_i\in T_i$. Furthermore, for any $F_i\in\T_i$,
	\begin{align*}
		\int_{F_i} h_i(t_i)p_i(dt_i)&=\int_{F_i}\int_{T_{-i}} f(t)q_{-i}(d t_{-i})p_i(dt_i) \\
		&=p(\pi_i^{-1}(F_i)) = p_i(F_i),
	\end{align*}
	implying that $h_i = 1$, $p_i$-a.e. In particular, we have $p_i(H_i)=1$.

	Next, we show that $\widetilde{\eta}_i(t_i,\csdot)$ and $\eta_i(t_i,\csdot)$ agree on $\T_{0,-i} \eqdef \bigotimes_{j=0,j \ne i}^n \T_{0,j}$ for each $t_i$ outside a $p_i$-null set. For any rectangle $F_{-i}=\prod_{j \ne i} F_j \in \T_{-i}$ and any $F_i \in \T_i$, let
	\begin{align*}
		\widetilde{F}_i \eqdef F_0 \times\cdots\times F_{i-1} \times (F_i \cap H_i) \times F_{i+1} \times \cdots \times F_n.
	\end{align*}
	Then, we have
	\begin{align}
		\label{eq:kernels_equality}
		\begin{aligned}
			\int_{F_i}\widetilde{\eta}_i(t_i, F_{-i})p_i(dt_i) & =\int_{F_i \cap H_i} \int_{F_{-i}} f(t)q_{-i}(d t_{-i})p_i(dt_i) \\
			& =\int_{\widetilde{F}_i} f(t)q(d t) =\int_{\widetilde{F}_i} p(d t) \\
			& =\int_{F_i \cap H_i}\eta_i(t_i, F_{-i}) p_i(d t_i)=\int_{F_i}\eta_i(t_i, F_{-i}) p_i(dt_i),
		\end{aligned}
	\end{align}
	where the first and the last equalities follow because $p_i(H_i^c) = 0$, the third equality follows from the definition of $f$, and the fourth equality follows from the consistency of $\eta$ under $p$. Since $\T_{0,j}$, $j=0,\ldots, n$, are countably generated, there exists a countable family of rectangles $\{R_k\}_{k\ge 1}$ in $T_{-i}$ of the form $\prod_{j=0,j\ne i}^n R_{k,j}$ that is closed under finite intersections and generates $\T_{0,-i}$. In light of \eqref{eq:kernels_equality}, for each $k\ge 1$, there exists $N_k\in \T_i$ with $p_i(N_k)=0$ such that for each $t_i\notin N_k$, $\widetilde{\eta}_i(t_i,R_k)=\eta_i(t_i,R_k)$. We conclude by the $\pi$-$\lambda$ theorem that $\widetilde{\eta}_i(t_i,\csdot)$ and $\eta_i(t_i,\csdot)$ agree on $\T_{0,-i}$ outside the $p_i$-null set $\bigcup_{k\ge 1}N_k$. Lemma \ref{lemm:measures_eq} in Section \ref{app2} below extends the agreement from $\T_{0,-i}$ to $\T_{-i}$.
\end{proof}

For the rest of this section, we assume that Assumptions \ref{assump:basic}, \ref{assump:payoff-1}, and \ref{assump:payoff-2} hold. Given a topological space $X$, let $\P(X)$ denote the set of all probability measures on $(X,\B(X))$, endowed with the topology of weak convergence, i.e., the coarsest topology that makes the map $\mu \mapsto \int f\,d\mu$ continuous for every bounded continuous function $f: X \to \R$. Define $T'\eqdef \prod_{j=1}^n T_j$ with the corresponding product $\sigma$-algebra $\T'\eqdef \bigotimes_{j=1}^n \T_j$. Similarly, for $i=1,\ldots,n$, let $T'_{-i}\eqdef \prod_{j=1,j\ne i}^n T_j$ so that $T_{-i}=T_0\times T'_{-i}$.

Given $\gamma_{-i}\in \RC_{-i}\eqdef \prod_{j\ne i}\RC_j$, we let $\kappa_i({\gamma}_{-i}): T'_{-i} \times \B(A_{-i}) \rightarrow [0,1]$ denote the tensor product of the elements of $\gamma_{-i}$, i.e., for any Borel rectangle $B= \prod_{j \ne i} B_j$, with $B_j\in \B(A_j)$,
\[
	\kappa_i({\gamma}_{-i})(t'_{-i}, B) = \prod_{j \ne i }\gamma_{j}(t_j,B_j).
\]
For each $i=1,\ldots,n$, define the probability kernel $\eta_i \otimes \kappa_i({\gamma}_{-i}):T_i \times (\T_{-i}\otimes \B(A_{-i}))\to [0,1]$ by
\[
	(\eta_i \otimes \kappa_i(\gamma_{-i}))(t_i,F\times B)=\int_F \kappa_i({\gamma}_{-i})(t'_{-i}, B)\,
	\eta_i(t_i, d (t_0,t'_{-i}))
\]
for any $F \in \T_{-i}$ and $B \in \B(A_{-i})$, and the map $g_i:T_i \times A_i\times \RC_{-i}\to \R$, given by
\[
	g_i(t_i,a_i,\gamma_{-i}) = \int u_i(a,t) (\eta_i \otimes \kappa_i({\gamma}_{-i}))(t_i, d(t_{-i},a_{-i})).
\]

The following lemmas provide the foundation for our main result, showing that any Nash equilibrium of a surrogate game can be modified into a Bayesian equilibrium of the game $G$.

\begin{lemma}
	\label{lemm:upper_semicontinuity}
	\begin{enumerate}[wide=0em,label=(\roman*)]
		\item For each $i=1,\ldots,n$, $t_i \in T_i$, and $\gamma_{-i} \in \RC_{-i}$, the map $a_i\mapsto g_i(t_i,a_i,{\gamma}_{-i})$ is bounded and upper semicontinuous on $A_i$.
		\item For each $i=1,\ldots, n$ and $\gamma_{-i} \in \RC_{-i}$, the map $(t_i,a_i)\mapsto g_i(t_i,a_i,\gamma_{-i})$ is $\T_i\otimes \B(A_i)$-measurable.
	\end{enumerate}
\end{lemma}

\begin{proof}
	To show upper semicontinuity in (i), fix $(t_i,a_i,\gamma_{-i}) \in T_i \times A_i\times \RC_{-i}$ and take a sequence $a_{i,k} \to a_i$ in $A_i$. Then, using the reverse Fatou lemma and Assumptions \ref{assump:payoff-1} and \ref{assump:payoff-2}, we find that
	\begin{align*}
		\limsup_{k \to \infty} g_i(t_i,a_{i,k},\gamma_{-i}) & \le \int \limsup_{k \to \infty} u_i(a_{i,k},a_{-i},t) (\eta_i \otimes \kappa_i({\gamma}_{-i}))(t_i, d(t_{-i},a_{-i})) \\
		& \le \int u_i(a,t) (\eta_i \otimes \kappa_i({\gamma}_{-i}))(t_i, d(t_{-i},a_{-i})) = g_i(t_i,a_i,{\gamma}_{-i}).
	\end{align*}
	The boundedness in (i) follows from Assumption \ref{assump:payoff-1}. Finally, (ii) follows from Lemma \ref{lemm:mble_integral1} in Section \ref{app2} below, where, by Assumption \ref{assump:payoff-1}, $D=T_i \times A_i$.
\end{proof}

For $r_i \in \P(A_i)$, we extend the notation $U_i(\gamma_i,\gamma_{-i};t_i)$ by writing $U_i(r_i, \gamma_{-i}; t_i)$ for the expected payoff when player $i$ of type $t_i$ plays the randomized action $r_i$.

\begin{lemma}
	\label{lemm:measurable_maximizer}
	For each $i=1,\ldots, n$ and $\gamma_{-i} \in \RC_{-i}$, the map $t_i\mapsto \sup_{r_i \in \P(A_i)} U_i(r_i,\gamma_{-i};t_i)$ is $\T_i$-measurable, and there exists a $(\T_i,\B(\P(A_i)))$-measurable map $\beta_i(\csdot;\gamma_{-i}): T_i \to \P(A_i)$ such that for each $t_i \in T_i$,
	\begin{align*}
		U_i(\beta_i(t_i;\gamma_{-i}),\gamma_{-i};t_i) = \sup_{r_i \in \P(A_i)} U_i(r_i,\gamma_{-i};t_i).
	\end{align*}
\end{lemma}

% Note: since A_i is compact and metrizable, it is separable. It is also complete, and thus Polish making it a Souslin space.

\begin{proof}
	Fix $\gamma_{-i} \in \RC_{-i}$ and define the map $\psi_i: T_i \times \P(A_i) \rightarrow \R$ by setting
	\begin{align*}
		\psi_i(t_i,r_i)=\int g_i(t_i,a_i,\gamma_{-i}) r_i(da_i) = U_i(r_i,\gamma_{-i};t_i).
	\end{align*}
	By Lemma \ref{lemm:mble_integral2} in Section \ref{app2} below and Lemma \ref{lemm:upper_semicontinuity}(ii), the map $\psi_i$ is $\T_i \otimes \B(\P(A_i))$-measurable. Next, $\P(A_i)$ is compact and metrizable by Theorem 15.11 of \cite{Aliprantis/Border:06:InfiniteDimensionalAnalysis}, and for a fixed $t_i\in T_i$, the map $r_i\mapsto\psi_i(t_i,r_i)$ is upper semicontinuous by Theorem 15.5 of \cite{Aliprantis/Border:06:InfiniteDimensionalAnalysis} and Lemma \ref{lemm:upper_semicontinuity}(i). Thus, by Lemma \ref{lemm:rieder} in Section \ref{app2} below, the map $t_i\mapsto \sup_{r_i \in \P(A_i)} \psi_i(t_i,r_i)$ is $\T_i$-measurable, and there exists a $(\T_i,\B(\P(A_i)))$-measurable maximizer $r_i^*:T_i\to \P(A_i)$.
	We set $\beta_i(\csdot;\gamma_{-i}) = r_i^*$.
\end{proof}

Consider $\nu\in \M$ and let $\widetilde{\gamma}^*\in \NE_\nu$. For each $i=1,\ldots,n$, define
\[
	\widetilde{T}_i \eqdef \left\{ t_i \in T_i : U_i (\widetilde{\gamma}_i^*,\widetilde{\gamma}_{-i}^*;t_i) = \sup_{r_i \in \P(A_i)} U_i (r_i,\widetilde{\gamma}_{-i}^*;t_i) \right\},
\]
and let $\beta_i(\csdot;\widetilde{\gamma}_{-i}^*): T_i \to \P(A_i)$ be as in Lemma \ref{lemm:measurable_maximizer}. Note that by Lemma \ref{lemm:measurable_maximizer}, $\widetilde{T}_i\in \T_i$. Next, for $i=1,\ldots, n$, $t_i \in T_i$, and $A'_i \in \B(A_i)$, we set
\begin{align}
	\label{eq:beq}
	\gamma_i^*(t_i,A'_i) \eqdef \widetilde{\gamma}_i^*(t_i,A'_i) \cdot 1\{t_i \in \widetilde{T}_i\} + \beta_i(t_i;\widetilde{\gamma}_{-i}^*)(A'_i) \cdot 1\{t_i \notin \widetilde{T}_i\}.
\end{align}
Since the map $t_i\mapsto\gamma_i^*(t_i, A_i')$ is $\T_i$-measurable \citep[see, e.g.,][Lemma 15.16]{Aliprantis/Border:06:InfiniteDimensionalAnalysis}, $\gamma_i^*\in \RC_i$. We show that the modified strategy is a $\nu_i$-version of $\widetilde{\gamma}_i^*$, and that $(\gamma_1^*,\ldots,\gamma_n^*)$ is a Bayesian equilibrium.

% Note: t_i -> β_i(t_i; ...)(A) is measurable as a composition of the map from Lemma 15.16 of A&B and β_i.

\begin{lemma}
	\label{lemm:version}
	For $i = 1,\ldots,n$, $\gamma_i^*$ is a $\nu_i$-version of $\widetilde{\gamma}_i^*$, and for each $t_i \in T_i$ and $r_i \in \P(A_i)$,
	\[
		U_i(r_i,\gamma_{-i}^*;t_i)=U_i(r_i, \widetilde{\gamma}_{-i}^*;t_i).
	\]
\end{lemma}

\begin{proof}
	Suppose that $\nu_i(\widetilde{T}_i^c) > 0.$ By Lemma \ref{lemm:measurable_maximizer}, there exists a measurable map $\beta_i(\csdot;\widetilde{\gamma}_{-i}^*): T_i \to \P(A_i)$ such that
	\begin{align*}
		\sup_{r_i \in \P(A_i)} U_i (r_i,\widetilde{\gamma}_{-i}^*;t_i) = U_i (\beta_i(t_i;\widetilde{\gamma}_{-i}^*),\widetilde{\gamma}_{-i}^*;t_i).
	\end{align*}
	Since $U_i (\beta_i(t_i;\widetilde{\gamma}_{-i}^*),\widetilde{\gamma}_{-i}^*;t_i) > U_i(\widetilde{\gamma}_i^*,\widetilde{\gamma}_{-i}^*;t_i)$ on $\widetilde{T}_i^c$, we find that
	\[
		\widetilde{U}_i(\beta_i(\csdot;\widetilde{\gamma}_{-i}^*),\widetilde{\gamma}_{-i}^*;\nu) > \widetilde{U}_i(\widetilde{\gamma}_i^*, \widetilde{\gamma}_{-i}^*;\nu),
	\]
	which contradicts $\widetilde{\gamma}^*\in \NE_\nu$. Hence, we must have $\nu_i(\widetilde{T}_i^c) = 0$.

	As for the second assertion, let $H_i \eqdef \bigcup_{j \ne i} \pi_{ij}^{-1}(\widetilde{T}_j^c)$, where $\pi_{ij}:T'_{-i} \rightarrow T_j$, $j\ne i$, denotes the coordinate projection. Since $\nu_j(\widetilde{T}_j^c) = 0$ for all $j$, we have $\nu_{-i}(T_0 \times H_i) = 0$ and, consequently, $\eta_i(t_i,T_0 \times H_i) = 0$ for all $t_i \in T_i$. Finally, since $\widetilde{\gamma}_j^*(t_j,\csdot) = \gamma_j^*(t_j,\csdot)$ on $\widetilde{T}_j$ for all $j$,
	\begin{align*}
		U_i(r_i, \widetilde{\gamma}_{-i}^*;t_i) &= \int g_i(t_i,a_i, \widetilde{\gamma}_{-i}^*) r_i(da_i) \\
		&= \int g_i(t_i,a_i, \gamma_{-i}^*) r_i(da_i) = U_i(r_i, \gamma_{-i}^*;t_i),
	\end{align*}
	as required.
\end{proof}

\begin{lemma}
	\label{lemm:BE}
	The profile of behavioral strategies $(\gamma_1^*,\ldots, \gamma_n^*)$ is a Bayesian equilibrium.
\end{lemma}

\begin{proof}
	As argued above, $\gamma^*\in \RC$. Thus,
	it suffices to show that for each $i=1,\ldots,n$ and $t_i\in T_i$,
	\[
		U_i(\gamma_i^*,\gamma_{-i}^*;t_i)= \sup_{r_i \in \P(A_i)} U_i(r_i,\gamma_{-i}^*;t_i).
	\]
	Take $i\in\{1,\ldots,n\}$. For $t_i \in \widetilde{T}_i$, we have
	\begin{align*}
		U_i(\gamma_i^*,\gamma_{-i}^*;t_i) & = U_i(\gamma_i^*,\widetilde{\gamma}_{-i}^*;t_i) = U_i(\widetilde{\gamma}_i^*,\widetilde{\gamma}_{-i}^*;t_i) \\
		& = \sup_{r_i \in \P(A_i)} U_i(r_i,\widetilde{\gamma}_{-i}^*;t_i) = \sup_{r_i \in \P(A_i)} U_i(r_i,\gamma_{-i}^*;t_i),
	\end{align*}
	where the first and last equalities follow from Lemma \ref{lemm:version}, while the second and third equalities follow from the definitions of $\gamma_i^*$ and $\widetilde{T}_i$, respectively. Similarly, for $t_i \notin \widetilde{T}_i$, we have
	\begin{align*}
		U_i(\gamma_i^*,\gamma_{-i}^*;t_i) & = U_i(\beta_i(t_i;\widetilde{\gamma}_{-i}^*),\gamma_{-i}^*;t_i) = U_i(\beta_i(t_i;\widetilde{\gamma}_{-i}^*),\widetilde{\gamma}_{-i}^*;t_i) \\
		&= \sup_{r_i \in \P(A_i)} U_i(r_i,\widetilde{\gamma}_{-i}^*;t_i) = \sup_{r_i \in \P(A_i)} U_i(r_i,\gamma_{-i}^*;t_i),
	\end{align*}
	where the second and fourth equalities follow from Lemma \ref{lemm:version}, and the third equality follows from Lemma \ref{lemm:measurable_maximizer}.
\end{proof}

\begin{proof}[Proof of Lemma \ref{lemm:BE_main}]
	Let $\gamma_i^*$ be the modification of $\widetilde{\gamma}_i^*$ defined in \eqref{eq:beq} for $i=1,\ldots,n$. The result then follows from Lemmas \ref{lemm:version} and \ref{lemm:BE}.
\end{proof}

\begin{proof}[Proof of Corollary \ref{cor:BE_existence_1}]
	By Lemma \ref{lemm:integrable_dominating}, we can choose $\nu\in\M$. By Theorem \ref{thm:existence}, it suffices to show that $\NE_\nu \ne \varnothing$. We construct a modified surrogate Bayesian game $\widetilde{G}^{**}(\nu)$ with a common prior as follows. For each player $i=1,\ldots,n$, the payoff function $u_i^*: A \times T' \rightarrow \R$ is given by
	\begin{align*}
		u_i^*(a,t') = \int \widetilde{u}_i\left(a,(t_0,t');\nu\right)\nu_0(dt_0),
	\end{align*}
	where, using Lemma \ref{lemm:density_version} in Section \ref{app2} below, we choose a version $f_{\nu,i}$ of $d(\nu_i \otimes \eta_i)/d (\nu_i \otimes \nu_{-i})$, taking values in $[0,\infty)$, such that $\int f_{\nu,i}(t_i,t_{-i})\nu_0(dt_0)\in [0,\infty)$ for each $t'\in T'$.

	We treat $t_i \in T_i$ as private information of player $i$, which is drawn from $\nu_i$ independently across the players, and use $\nu'\eqdef \bigotimes_{j=1}^n \nu_j$ as a common prior. Then, $\NE_\nu \ne \varnothing$ if and only if the game $\widetilde{G}^{**}(\nu)$ has a Bayesian Nash equilibrium in behavioral strategies. To show the latter, we use Theorem 3.1 of \cite{Balder:88:MOR}. For this, we check conditions $C_1$--$C_3$ of the theorem, endowing $\RC = \RC_1 \times \cdots \times \RC_n$ with the product topology of the narrow topologies.\footnote{
		Note that the topology on the set of probability kernels constructed and called a ``weak topology'' by \cite{Balder:88:MOR} is also called a narrow topology in the literature \citep[see][Ch.~3]{Crauel:02:RandomProbabilityMeasures}. \cite{Balder:88:MOR} does not require any topology on $T_i$. Hence, the narrow topology of his paper carries over to our setting.
	}

	We verify $C_1$ first. Note that the payoff function $u_i^*$ is $(\B(A)\otimes \T')$-measurable by Lemma \ref{lemm:mble_integral1} in Section \ref{app2} below, where we treat $\nu_0$ as a constant probability kernel; thus, $C_1'$ is satisfied. For $C_1''$, the continuity of $u_i^*$ on $A$ for each fixed $t'\in T'$ follows by the dominated convergence theorem with $\overline{u}_i(t_i) f_{\nu,i}(t_i,t_{-i})$ as the dominating function. Condition $C_1'''$ holds with $\varphi_i:T'\to [0,\infty)$ given by
	\[
		\varphi_i(t')=\overline{u}_i(t_i)\int f_{\nu,i}(t_i,t_{-i})\nu_0(dt_0).
	\]
	Indeed, $|u_i^*| \le \varphi_i$ on $A\times T'$, and since $\nu\in\M$, we have
	\begin{align*}
		\int_{T'} \varphi_i(t')\nu'(dt')&=\int_{T'} \overline{u}_i(t_i)\left(\int_{T_0}f_{\nu,i}(t_i,t_{-i})\nu_0(dt_0)\right)\nu'(dt') \\
		&=\int_{T_i} \overline{u}_i(t_i)\left(\int_{T_{-i}}f_{\nu,i}(t_i,t_{-i})\nu_{-i}(d t_{-i})\right)\nu_i(dt_i) \\
		&=\int_{T_i}\overline{u}_i(t_i)\eta_i(t_i,T_{-i})\nu_i(dt_i)=\int_{T_i} \overline{u}_i(t_i)\nu_i(dt_i)<\infty.
	\end{align*}
	Finally, conditions $C_2$ and $C_3$ hold by construction and Assumption \ref{assump:basic}.
\end{proof}

\section{Auxiliary Results}
\label{app2}

Given a $\sigma$-algebra $\X$, $\overline{\X}$ denotes its universal completion. We let $\overline{\R}\eqdef \R\cup \{-\infty,\infty\}$.

\begin{lemma}
	\label{lemm:measures_eq}
	Let $(X_1,\X_1),\ldots, (X_m,\X_m)$ be measurable spaces and let $X\eqdef \prod_{j=1}^m X_j$, $\X\eqdef \bigotimes_{j=1}^m \X_j$, and $\widetilde{\X}\eqdef \bigotimes_{j=1}^m \overline{\X}_j$. Suppose that $\mu$ and $\mu'$ are probability measures on $(X,\X')$, where $\X\subset\X'\subset \widetilde{\X}$. Then, $\mu=\mu'$ if and only if $\mu(B)=\mu'(B)$ for each $B\in \X$.
\end{lemma}

\begin{proof}
	The necessity is immediate. For sufficiency, suppose that $\mu$ and $\mu'$ agree on $\X$. Consider any $C\in \X'$. Since $\X' \subset \X^{\mu}$, the $\mu|_{\X}$-completion of $\X$, we have $C\in \X^{\mu}$ and there exist $B,N\in \X$ such that $B\subset C\subset B\cup N$ and $\mu(N)=0$. Then, $\mu'(N)=0$, and consequently,
	\[
		\mu(C)=\mu(B)=\mu'(B)=\mu'(C),
	\]
	as required.
\end{proof}

\begin{lemma}
	\label{lemm:density_version}
	Let $(X,\X,\mu)$ and $(Y,\Y,\nu)$ be probability spaces and let $\lambda$ be a probability measure on $(X\times Y,\X\otimes \Y)$ such that $\lambda\ll \mu\otimes \nu$. Then, there exists a version $f$ of $d\lambda/d(\mu \otimes \nu)$ such that $f$ takes values in $[0,\infty)$, and
	\[
		h(x)\eqdef \int_Y f(x,y)\nu(dy)\in [0,\infty)
	\]
	for each $x\in X$.
\end{lemma}

\begin{proof}
	First, choose a version $f$ of $d\lambda/d(\mu\otimes\nu)$ such that $f$ takes values in $[0,\infty)$, making $h\ge 0$. Applying the Fubini--Tonelli theorem, we get
	\[
		\int h d\mu=\int f d(\mu \otimes \nu)=\lambda(X\times Y)=1.
	\]
	Therefore, $h<\infty$, $\mu$-a.e. Let $N\eqdef \{h=\infty\}\times Y$. Since $(\mu \otimes \nu)(N)=0$, we can set $f=0$ on $N$. With the redefined version of $f$, $h(x)\in [0,\infty)$ for each $x\in X$.
\end{proof}

\begin{lemma}
	\label{lemm:mble_integral1}
	Let $(X,\X)$ and $(Y,\Y)$ be measurable spaces and let $\rho:X\times \Y\to[0,1]$ be a probability kernel. Suppose further that $f:X\times Y\to \R$ is $(\X\otimes \Y)$-measurable. Define the map $g:D\to \overline{\R}$ by
	\[
		g(x)= \int_{Y} f(x,y)\rho(x,dy),
	\]
	where $D\eqdef\left\{x\in X: \int f(x,y)\rho(x,dy) \text{ exists}\right\}$. Then, $D\in \X$ and $g$ is $(\X|_D,\B(\overline{\R}))$-measurable.
\end{lemma}

\begin{proof}
	Let $\H$ be the class of bounded $(\X\otimes \Y)$-measurable functions $h$ such that the map $x\mapsto \int h(x,y)\rho(x,dy)$ is measurable. First, $\H$ contains indicators of measurable rectangles because for $B\in \X$ and $C\in \Y$,
	\[
		\int 1_B(x)1_C(y)\rho(x,dy)=1_B(x)\rho(x,C)
	\]
	is $\X$-measurable. Clearly, $\H$ is a vector space, closed under bounded monotone increasing limits of nonnegative functions in $\H$. Hence, by the functional monotone class theorem \citep[see, e.g.,][Theorem~5.2.2]{Durrett:Probability}, $\H$ contains all bounded $(\X\otimes \Y)$-measurable functions.

	Write $f=f^{+} - f^{-}$, where $f^{+}\eqdef f\vee 0$ and $f^{-}\eqdef -(f\wedge 0)$. By the previous result and the monotone convergence theorem, both
	\[
		g^{+}(x)\eqdef\int f^{+}(x,y)\rho(x,dy) \qtext{and}\quad g^{-}(x)\eqdef\int f^{-}(x,y)\rho(x,dy)
	\]
	are $(\X,\B(\overline{\R}))$-measurable. Hence, $D=\{x\in X: g^{+}(x)\wedge g^{-}(x)<\infty\}\in \X$. On this set, $g(x)=g^{+}(x)-g^{-}(x)$, which is $(\X|_D,\B(\overline{\R}))$-measurable.
\end{proof}

\begin{lemma}
	\label{lemm:mble_integral2}
	Let $(X,\X)$ be a measurable space and let $Y$ be a metrizable space. Suppose further that $f: X \times Y \to \R$ is $(\X\otimes \B(Y))$-measurable. Define the map $g: D \to \overline{\R}$ by
	\[
		g(x, r) = \int_Y f(x,y)r(dy),
	\]
	where $D\eqdef\left\{(x,r)\in X\times \P(Y): \int f(x,y)r(dy) \text{ exists}\right\}$. Then, $D\in \X\otimes \B(\P(Y))$ and $g$ is $(\X\otimes \B(\P(Y))|_D,\B(\overline{\R}))$-measurable.
\end{lemma}

\enlargethispage{1em}

\begin{proof}
	Consider the projection $\pi:X\times \P(Y)\times Y\to X\times Y$ given by $\pi(x,r,y)=(x,y)$. We extend $f$ to $\widetilde{f}$ on $X \times \P(Y) \times Y$ by setting $\widetilde{f} \eqdef f\circ \pi.$ Note that $\widetilde{f}$ is $(\X \otimes \B(\P(Y))\otimes\B(Y))$-measurable as a composition of measurable functions. For each $x \in X$, $r \in \P(Y)$, and $B \in \B(Y)$, define $\rho(x,r,B) \eqdef r(B)$. Then, $\rho$ forms an $(\X\otimes \B(\P(Y)))$-measurable probability kernel, $\rho:(X \times \P(Y)) \times \B(Y) \to [0,1]$, because for any Borel $B$, the evaluation map $r \mapsto r(B)$ is Borel \citep[see, e.g.,][Lemma 15.16]{Aliprantis/Border:06:InfiniteDimensionalAnalysis}. Note that
	\[
		\int f(x,y)r(dy) = \int\widetilde{f}(x,r,y) \rho(x,r,dy),
	\]
	when the integrals exist. Hence, the result follows from Lemma \ref{lemm:mble_integral1}.
\end{proof}

\begin{lemma}
	\label{lemm:rieder}
	Let $(X,\X)$ be a measurable space such that either \tn{(a)} $X$ is a Borel space with $\X = \B(X)$, or \tn{(b)} $\X$ is universally complete. Let $Y$ be a compact metrizable space and let $f: X\times Y \to \R$ be an $\X\otimes\B(Y)$-measurable function such that $y\mapsto f(x,y)$ is upper semicontinuous on $Y$ for each $x \in X$. Then, the map $x \mapsto \sup_{y \in Y}f(x,y)$ is $\X$-measurable, and there exists an $(\X,\B(Y))$-measurable map $y^*: X \to Y$ such that
	\[
		f(x,y^*(x)) = \sup_{y\in Y} f(x,y), \quad x \in X.
	\]
\end{lemma}

Lemma \ref{lemm:rieder} is a special case of Theorem 3.5 of \cite{Rieder:78:ManuscriptaMath}, with its condition (iii) in the form given by Corollary 3.7 therein. The sets $\{y \in Y: f(x,y)\ge c\}$, $c \in \R$, are compact by the upper semicontinuity of $y\mapsto f(x,y)$ and compactness of $Y$, and the required selection classes for (a) and (b) are provided by Example 2.6 and Example 2.3 of that paper, respectively.

Let $\xi$ follow a normal distribution with location $\mu \in \R$ and scale $\sigma > 0$, truncated to the interval $[a,b]$ with $-\infty < a < b < \infty$. For $x\in\R$, let $\tau(x;\mu)\eqdef (x-\mu)/\sigma$. Define $D(\mu) \eqdef \Phi(\tau(b;\mu)) - \Phi(\tau(a;\mu))$, where $\phi$ and $\Phi$ are the PDF and CDF of the standard normal distribution. The PDF and CDF of $\xi$ on $[a,b]$ are given by
\[
	f(x;\mu) = \frac{\phi(\tau(x;\mu))}{\sigma D(\mu)}
	\qtext{and}\quad
	F(x;\mu) = \frac{\Phi(\tau(x;\mu)) - \Phi(\tau(a;\mu))}{D(\mu)},
\]
respectively. (For brevity, we have suppressed the dependence on $\sigma$.)

\begin{lemma}
	\label{lemm:trunc_normal_drv}
	For every $x \in [a,b]$ and $\mu\in \R$, $-(2\sigma)^{-1} \le \partial_{\mu} F(x;\mu) \le 0$.
\end{lemma}

\begin{proof}
	For $x \in \{a,b\}$, $F(x;\mu)$ equals $0$ or $1$ for every $\mu$, and the bounds hold trivially. Fix $x\in (a,b)$. A direct computation gives $\partial_\mu \log f(x;\mu) = (x - \E[\xi])/\sigma^2$, where $\E[\xi]$ is the expected value of $\xi$. Differentiating $F(x;\mu) = \int_a^x f(z;\mu)\,dz$ under the integral sign, we get
	\begin{align*}
		\partial_{\mu} F(x;\mu) &= \int_a^x f(z;\mu)\cdot \partial_\mu \log f(z;\mu)\,dz \\
		&=\int_a^x f(z;\mu)\cdot \frac{z-\E[\xi]}{\sigma^2} \,dz
		= \frac{1}{\sigma^2}\Cov(\xi,1\{\xi \le x\}).
	\end{align*}
	Since $\E[\xi \mid \xi \le x] \le \E[\xi \mid \xi > x]$,
	\[
		\Cov(\xi, 1\{\xi \le x\}) = F(x;\mu)(1 - F(x;\mu))(\E[\xi \mid \xi \le x] - \E[\xi \mid \xi > x]) \le 0,
	\]
	implying $\partial_\mu F(x;\mu) \le 0$. For the lower bound, the Cauchy--Schwarz inequality yields
	\[
		|\Cov(\xi, 1\{\xi \le x\})|
		\le \sqrt{\Var(\xi)\Var(1\{\xi \le x\})}
		\le \sigma/2
	\]
	because $\Var(1\{\xi \le x\}) = F(x;\mu)(1 - F(x;\mu)) \le 1/4$ and $\Var(\xi) \le \sigma^2$. Hence, $\partial_\mu F(x;\mu) \ge -\sigma/(2\sigma^2) = -1/(2\sigma)$.
\end{proof}

\bookmarksetup{startatroot}

\bibliographystyle{elsarticle-harv}
\bibliography{inconsistent_beliefs_BE}

\end{document}